\documentclass[11pt]{amsart}
\usepackage{amssymb,amscd,amsmath,amsthm}
\usepackage{graphicx}
\graphicspath{ {./images/} }

\newtheorem{theorem}{Theorem}[section]

\newtheorem{lemma}[theorem]{Lemma}
\newtheorem{definition}[theorem]{Definition}

\def\bh{{\mathbb H}}

\def\b{\beta}

\def\Tr{\mathrm{Tr}}
\def\<{\langle}
\def\>{\rangle}

\def\1{\mathbf{1}}

\def\bh{\mathbf{h}}

\def\id{{\bf 1}\!\!{\rm I}}
\def\Tr{\mathrm{Tr}}
\begin{document}

UDK 517.98\\

\begin{center}
{\Large {\bf Quantum Markov Chains on the Comb graphs: Ising model}}\\[1cm]

{\sc Farrukh Mukhamedov} \\[2mm]

 Department of Mathematical Sciences, College of Science, \\
 United Arab Emirates University 15551, Al-Ain,\\ United
Arab Emirates.\\

e-mail: {\tt far75m@gmail.com; farrukh.m@uaeu.ac.ae}\\[1cm]

{\sc Abdessatar Souissi} \\[2mm]

Department of Accounting, College of Business Management\\
Qassim University, Ar Rass, Saudi Arabia \, and \\
Preparatory institute for scientific and technical studies,\\
 Carthage University, Amilcar 1054, Tunisia\\
  e-mail: {{\tt a.souaissi@qu.edu.sa}; {\tt
 abdessattar.souissi@ipest.rnu.tn}}\\[1cm]

 {\sc Tarek Hamdi} \\[2mm]

Department of Management Information Systems, College of Business Management\\
Qassim University, Ar Rass, Saudi Arabia \, and \\
Laboratoire d'Analyse Math\'ematiques et applications LR11ES11 \\
Universit\'e de Tunis El-Manar, Tunisia\\
  e-mail: {{\tt t.hamdi@qu.edu.sa}}\\[1cm]
\end{center}

\small
\begin{center}
{\bf Abstract}\\
\end{center}
In the present paper, we construct quantum Markov chains (QMC) over the Comb graphs. As an application of this construction, it is proved the existence of the disordered phase for the Ising type models (within QMC scheme) over the Comb graphs. Moreover, it is also established that
the associated QMC has clustering property with respect to translations of the graph. We stress that this paper is the first one where a nontrivial example of QMC over non-regular graphs is given.

\normalsize

\section{Introduction}

Over the past decade, motivated
largely by the prospect of superefficient algorithms, the theory of quantum Markov
chains (QMC), especially in the guise of quantum walks, has generated a huge number of works,
including many discoveries of fundamental importance \cite{AW,DM2019,Ke,Pr}.
In \cite{Fing} it has been proposed a
novel approach to investigate quantum cryptography problems by means of QMC \cite{Gud} where quantum effects are entirely encoded into super-operators labelling transitions, and
the nodes of its transition graph carry only classical information and thus they are discrete. Recently, QMC have been applied  \cite{DK19,DM2019,1DM2019} to the investigations of so-called "open quantum random walks" \cite{attal,CP2,konno}.

On the other hand, in physics, a spacial classes of QMC, called "Matrix Product States" (MPS) and more generally "Tensor Network States" \cite{CV,Or} were used to investigate
quantum phase transitions for several lattice models.
This method uses the density matrix renormalization group (DMRG) algorithm which opened a new way of performing
the renormalization  procedure in 1D systems and gave extraordinary precise results. This is done by keeping the states of subsystems which
are relevant to describe the whole wave-function, and not those that minimize the energy on
that subsystem.
Those states had appeared in the
literature in many different contexts and with different names:
\begin{itemize}
\item[(a)] variational ansatz for
the transfermatrix in the estimation of the partition function of a classical model \cite{KW},;
\item[(b)] in the AKLT model in 1D  \cite{aklt}, where the ground state has the form of a valence bond
solid (VBS) \cite{LPS} which can be exactly written as an MPS. Translational invariant MPS in infinite
chains were thoroughly studied and characterized mathematically in full generality in  \cite{fannes, fannes2},
where they are called as finitely correlated states (FCS) (see also \cite{[Ac74f],[AcFri]} for closely related paper where QMC approach was used).
\end{itemize}

In \cite{AMSa1,AMSa2,AMSa3,AOM} it has been used a QMC approach to investigate models defined over the Cayley trees.
In this path the QMC scheme is based on the
 $C^*$-algebraic framework (see also \cite{ASG20}).
Furthermore, in \cite{MBS161,MBS162,MBSG20,MR1,MR2,MuSou2018} we have established that Gibbs measures of the Ising model with competing (Ising) interactions (with commuting interactions) on a Cayley trees,  can be considered as  QMC. Note that if the perturbation vanishes then the model reduces to the classical Ising one which was also examined in  \cite{ArE} by means of $C^*$-algebraic methods.
 Other types of models with $XY$-interactions on the same tree have studied in \cite{MG17,MG19,MG20}.
In \cite{NFG} using the
 matrix product states, it has been numerically investigated the quantum Ising model in a transverse field on the Cayley tree.
However, the Cayley tree is considered as a regular graph, therefore, it is interesting to develop QMC scheme for non-regular graphs.
One of the simplest non-regular graph is the Comb graph which has many applications in computer sciences. Several physical models over such
graph were investigated in \cite{Bax}.

The main aim of the present paper is to construct QMC over the Comb graphs. We stress that the construction essentially uses boundary conditions
associated with density amplitudes. As an application of such kind of construction, the existence of the disordered phase of the Ising model within QMC scheme over the Comb graph is proved. Moreover, it is also established that
the associated QMC has clustering property with respect to translations of the graph. We point out that our paper is the first one where a construction of nontrivial example of QMC over the non-regular graph is given. Besides, the provided construction will allow to
investigate phase transition problem for other kinds of models over non-regular graphs.

\section{Comb graphs}\label{Combsection}

In this section, we recall some necessary notions about the Comb graphs.
Let $\mathcal  G   = (L,E)$  be a locally finite    connected graph with infinite set of vertices.
An edge $l\in E$ is associated to the  pair of its endpoints $l=<x,y> = <y,x>$ and $E$ is then identifiable to a subset of $L\times L$.
$$
E\subset \{\{x,y\} \; : \; x,y\in L\}.
$$
Recall that:
\begin{itemize}
\item Two vertices $x$ and $y$ are called {\it nearest neighbors}, and we denote by  $x\sim y$, if there exists an edge joining them.\\
\item A collection of the elements  $x\sim x_1\sim \dots\sim x_{d-1}\sim y>$ is called a {\it path} from the vertex $x$ to the vertex $y$ with length $d$.\\
\item The distance $d(x,y), x,y\in V$, on the Comb graph, is the length of the shortest path from $x$ to $y$.\\
\item The interaction domains at a given vertex $x$ is given by
\begin{equation}
N_x =  \bigl\{ y \in L \quad : \quad x\sim y \bigr\}
\end{equation}
 its cardinal $|N_x|$ is called the  \textbf{valence} at the site $x$.
\end{itemize}

In  the sequel, we deal with the Comb graph $\mathcal C_d : = \mathbb N^{\rhd d}$, which is a tree base

\begin{figure}[h]
\centering
\includegraphics[width=0.5\textwidth]{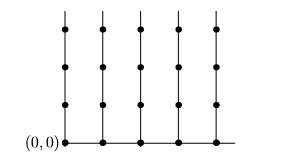}
\caption{Comb graph: $\mathbb N \rhd \mathbb N$}
\end{figure}

A  natural coordinate structure is associated to the Comb graph $\mathcal C$, as subset of the integer lattice $\mathbb N^d$. The vertex set is then $E=  \mathbb N^d$ and the  root is the origin $x^{(0)} := (0, \cdots, 0)$.
Each vertex $x$ is identified to a d-tuple   of  integers  $x\equiv (j_1, \cdots, j_d )$.
For each $k\in \{1, 2, \cdots, d\}$ consider
$$
e_k := (0, \cdots, 0 , 1 , 0, \cdots , 0)
$$
 where $1$ appears at the $k^{th}$-component.
 Let us set
\begin{eqnarray*}
 W_n &=&\bigg \{ x\in L \, : \, d(x,x^{(0)}) = n\bigg\}\\[2mm]
&=&
\bigg\{x=  (j_1, \cdots, j_d )\in \mathbb N^d \, : \, \sum_{k=1}^{d} j_k  =1\bigg\}.
\end{eqnarray*}
 \[   \Lambda_{n} =
\bigcup_{k=0}^n  W_k, \qquad   \Lambda_{[n,m]}=\bigcup_{k=n}^m  W_k, \
(n<m).
\]
Now for each $x\in  W_n$, we define its direct successors by
$$
S(x) : = \bigl\{ y\in  W_{n+1} \, : \,  x\sim y\bigr\}
$$
It is obvious that for  every vertex $x$ (distinct of the root $x^{(0)}$) one has  $|S(x)| = |N_x|-1 $.
In particular,  if  $x\in \Delta_k$ we get $|S(x)|  =  2(d-k ) + 1$.
Moreover, the tree structure of the considered graph yields to
\begin{equation}\label{Ln-union_S(x)}
 W_{n+1} = \bigcup_{x\in  W_n}S(x)  \quad \hbox{and} \quad S(x)\cap S(y) = \varnothing, \quad  \forall x\neq y.
\end{equation}

\section{Quantum Markov chains on the Comb graph}

Let $\mathcal A, \mathcal B$ be two unital C$^\ast$-algebras.  A  completely positive (CP) identity preserving map $\mathcal E : \mathcal A \otimes \mathcal B \to \mathcal B$ is called \textit{transition expectation}.
 Note the structure of completely positive maps  was clarified \cite{[Ac81]}  and its general formulation in terms of density matrices was carried out in the case of full matrix algebras.  Mainly, every CP-map $\mathcal E : \mathcal A \otimes \mathcal B \to \mathcal B $ is given by
$$
\mathcal E (x) = \sum_{j}\Tr_{\mathcal B]}\bigl( K_j^\ast x K_j); \quad x\in \mathcal A\otimes \mathcal B
$$
 where $K_j\in \mathcal A \otimes \mathcal B$ and $\Tr_{\mathcal B} : \mathcal A \otimes \mathcal B \to \mathcal B$ is the partial trace. In particular, if $K\in \mathcal A \otimes \mathcal B$ is any density, then
\begin{equation}
\mathcal E(x) = \Tr_{\mathcal B]}(K^\ast x K); \quad x\in \mathcal A \otimes \mathcal B
\end{equation}
such that
\begin{equation}\label{id_preser_E}
  \Tr_{\mathcal B]}(K^\ast  K) = I_{\mathcal B}
\end{equation}
is a (identity preserving) transition expectation.
In what follows, any operator $K\in \mathcal A\otimes \mathcal B$ which satisfies (\ref{id_preser_E}) will be called \textit{conditional density matrix.}

Consider a triplet $\mathcal C\subseteq\mathcal B\subseteq \mathcal A$ of C$^\ast$--algebras. A \textit{quasi-conditional expectation} \cite{ACe} is a completely positive identity preserving linear map $E :\mathcal A \to \mathcal B$ such that
$
E(ca) = cE(a)$, for all $a\in\mathcal A$, $c\in \mathcal C.
$

To each vertex $x$, we associate an algebra of observable  $\mathcal B_x\equiv \mathcal B(\mathcal H_x)$ for some finite  dimensional Hilbert space $\mathcal H_x$.  Consider the quasi-local algebra $\mathcal B_L : = \bigotimes_{x\in L} \mathcal B_x$ which the inductive limit of the net
$$
\mathcal B_\Lambda : = \bigotimes_{x\in \Lambda}\mathcal B_x \otimes I_{\Lambda^c}; \quad \Lambda\subset L, \,  |\Lambda|<\infty.
$$
Then
$$
\mathcal{B}_L = \overline{\mathcal{B}_{L, loc}}^{C^\ast}
$$
where
$$
\mathcal{B}_{L, loc} = \bigcup_{\Lambda \subset_{fin} L}\mathcal{B}_V.
$$
For the sake of simplicity, we will denote
$$
\mathcal{B}_{[m,n]} := \mathcal{B}_{\Lambda_{[m,n]}}.
$$

Starting from any quasi-conditional expectation $E_{n}: \mathcal B_{[0,n+1]} \to \mathcal B_{[0,n]}$ with respect to the triplet $\mathcal B_{ [0,n-1]} \subseteq \mathcal B_{ [0,n-1]}\subseteq\mathcal B_{ [0,n]}$ one can obtain a transition expectation from $\mathcal B_{[n,n+1]}$ into $\mathcal B_{ W_{n}}$  by the mere restriction
$$
 \mathcal E_{ [n,n+1]} :=  E_{n} \Big|_{\mathcal B_{[n, n+1]}}.
$$
Conversely, every transition expectation $\mathcal E_{[n,n+1]}:  \mathcal B_{[n,n+1]}\to  B_{ W_{n}}$ is extendable to  a quasi-conditional expectation  $E_{n}$ with respect to the above triplet in the following way
\begin{equation}\label{TE_QCE}
E_{n}:= id_{\mathcal B_{n-1]}}\otimes \mathcal E_{[n,n+1]}.
\end{equation}

 \begin{definition}\label{bQMC_def}
 A state $\varphi$ on $\mathcal B_L$  is called \textbf{quantum Markov chain} (QMC) with respect to a triplet $\big(\rho_0 , \{\mathcal E_{[n, n+1]}\}, \{h_n\}\big)$, where $\rho_0 $ is a state on $\mathcal B_{x^{(0)}}$, $\mathcal E_{[n ,n+1]} : \mathcal B_{[n,n+1]} \to \mathcal B_{ W_n}$ is a transition expectation and $\{h_n\}\subset\mathcal B_{ W_n, +}$ is a sequence of \textbf{boundary conditions}) if
\begin{equation}\label{bQMC_eq}
\varphi= \lim_{n \to +\infty} \rho_0\circ E_{0}\circ E_{1}\circ \cdots\circ E_{n}\circ {\mathbf{h}}_{n+1}
\end{equation}
where ${\mathbf{h}}_{k}(\cdot)=h_{k}^{1/2}(\cdot) h_k^{1/2}$.  Here we have used the notation \eqref{TE_QCE}.
\end{definition}

Thanks to \eqref{TE_QCE}, we can immediately verify that, the QMC $\varphi$ is evaluated at $a=a_0\otimes a_1 \otimes \dots \otimes a_{ W_n} \in \mathcal B_{ [0,n]}$, $a_j \in\mathcal B_{ W_j}$, by
$$
\varphi(a) =\lim_{n \to +\infty} \rho_0\Bigl(\mathcal E_{[0,1]}\Bigl(a_{ W_0}\otimes\mathcal E_{[1,2]}\Bigl(a_1\otimes\cdots \otimes \mathcal E_{[n ,n+1]}\Bigl(h^{1/2}_{n+1} a_n h^{1/2}_{n+1}\otimes I\Bigr)\Bigr)\Bigr)\Bigr)
$$
which highlights the quantum Markov chain structure.

We notice that a more general definition of QMC has been formulated in \cite{[AcFiMu07],{ASG20}}.

%
%

\section{Construction of QMC  on $\mathbb N^{\rhd d}$}

This section is devoted to a construction of quantum Markov chains associated with nearest-neighbors interactions on the algebra $\mathcal B_L$.

Let $K_{[n,n+1]} \in \mathcal B_{ [n,n+1]}$ be a conditional density matrix and let $\mathcal E_{[n ,n+1]} : \mathcal B_{ [n,n+1]} \to \mathcal B_{ W_n}$ its associated transition expectation given by
\begin{equation}\label{E_K}
\mathcal E_{[n ,n+1]}(a) = \Tr_{\mathcal B_{ W_n}]}\bigl(K_{[n, n+1]}^{\ast}aK_{[n, n+1]}\bigr)
\end{equation}
Note that  \eqref{Ln-union_S(x)} leads to  the identity
$$
\mathcal B_{ [n,n+1]} \equiv \bigotimes_{x\in  W_n}\mathcal B_{x\cup S(x)}
$$
A transition expectation $\mathcal E_{[n,n+1] } : \mathcal B_{ [n,n+1]} \to \mathcal B_{ W_n}$ is said to be \textit{localized} if it can be decomposed as follows
\begin{equation}\label{localized_E}
\mathcal E_{[n,n+1] } = \bigotimes_{x\in  W_n} \mathcal E_x
\end{equation}
for some transition expectation $\mathcal E_x : \mathcal B_{\{x\}\cup S(x)} \to \mathcal B_x$   for each vertex $x\in  W_n$.

Note that the partial trace $\Tr_{\mathcal B_{ W_n}]}$ is  by construction localized
$$
\Tr_{\mathcal B_{ W_n}]} = \bigotimes_{x\in  W_n}\Tr_{\mathcal B_{\{x\}\cup S(x)}]}
$$
One can see that the  localized transition expectation $\mathcal E_{[n,n+1]}$ highlights the fine structure of the considered graph. This property played an essential role in the study of phase transitions for quantum Markov chains on the Cayley tree \cite{MuSou2018}.

\begin{definition}
A conditional density  operator $K_{[n, n+1]}\in \mathcal B_{ [n,n+1]}$ is said to be \textit{localized} if
\begin{equation}\label{localized_K}
K_{[n,n+1]} = \bigotimes_{x\in  W_n}K_{\{x\}\cup S(x)}
\end{equation}
where $K_{\{x\}\cup S(x)}\in \mathcal B_{\{x\}\cup S(x)}$ for every $x\in  W_n$.
\end{definition}

We immediately can prove the following fact.

\begin{lemma}\label{localE_K_lem}
If the conditional density matrix $K_{[n,n+1]} =  \bigotimes_{x\in  W_n}K_{\{x\}\cup S(x)}$ is localized then the associated quasi-conditional expectation $\mathcal E_{[n, n+1]}$ is localized in the sense of (\ref{localized_E}) with
 \begin{equation}\label{Ex}
\mathcal E_x(\cdot)  = \Tr_{\mathcal B_x} \Bigl(K_{\{x\}\cup S(x) }^{\ast}\cdot K_{\{x\}\cup S(x)}  \Bigr).
\end{equation}
for each $x\in W_n$.
\end{lemma}



In the sequel, the pair $\{\rho_0,  h =(h_n)\}$ denotes a boundary condition such that $\rho_0$ is a faithful positive linear functional on the algebra $\mathcal B_{(x_0)}$ and  $h_n\in \mathcal B_{ W_n}^{+}$  for each $n$.

Let $\varphi^{(\rho_0, h)}_n$ be a linear functional on $\mathcal B_{[0,n]}$ given by
\begin{equation}\label{FF}
\varphi^{(\rho_0, h)}_n(a) = \rho_0\bigl(E_{0} \bigl( E_{1}\bigl( \cdots\bigl(  E_{n}\bigl(h^{1/2}_{n+1}ah^{1/2}_{n+1}\bigr)\bigr)\bigr)\bigr)\bigr).
\end{equation}

Observe that for each $a\in \mathcal B_{ [0,n]}$, one has
$$
E_n\bigl(h_{n+1}^{1/2}(a\otimes I)^\ast(a\otimes I)h_{n+1}^{1/2}\bigr)\ge 0
$$
and since the maps $\mathcal E_{[j, j+1]}$ are completely positive and the functional $\rho_0$ is a positive functional then  the functional $\varphi^{(\rho_0, h)}_n$ is positive.

The sequence $\{\varphi^{(\rho_0, h)}_n\}_n$ satisfies the \textit{compatibility condition} if
\begin{equation}\label{compatibility_eq}
\varphi_{n+1}^{(\rho_0, h)}\Big|_{\mathcal B_{ [0,n]}} = \varphi^{(\rho_0, h)}_n, \quad \ \textrm{for all} \  n.
\end{equation}
Clearly, the compatibility condition ensures the existence of the weak-limit
$$
\lim_{n\to\infty}\varphi^{(\rho_0, h)}_n.
$$

\begin{lemma}\label{E(hn+1)=hn} For the same notations as above, if
\begin{equation}\label{comp_E_h}
\mathcal E_{[n , n+1]}(h_{n+1}) = h_n\quad ; \quad \forall n\in\mathbb{N}
\end{equation}
then the  sequence $\{\varphi^{(\rho_0, h)}_n\}_n$ given by \eqref{FF} satisfies the compatibility condition \eqref{compatibility_eq}.

\end{lemma}
\begin{proof}
Let $a\in \mathcal B_{  \Lambda_{n}}$. Due to $[h_{n+i},a]=0$, ($i=1,2$) we find
\begin{eqnarray*}
\varphi_{n+1}^{(\rho_0, h)}(a) &=& \rho_0\bigl( \mathcal E_{[0 ,1]} \bigl( \mathcal E_{[1 ,2]}\bigl( \cdots \mathcal E_{[n  ,n+1]}\bigl(  \mathcal E_{[n+1 ,n+2]}\bigl(h_{n+2}a\otimes I\bigr)\bigr)\cdots \bigr)\bigr)\bigr)\bigr),\\[2mm]
&=&\rho_0\bigl( \mathcal E_{[0 ,1]} \bigl( \mathcal E_{[1 ,2]}\bigl( \cdots \mathcal E_{[n  ,n+1]}\bigl(  \mathcal E_{[n+1 ,n+2]}\bigl(h_{n+2}\bigr)a\otimes I\bigr)\cdots \bigr)\bigr)\bigr)\bigr),\\[2mm]
&=&\rho_0\bigl( \mathcal E_{[0 ,1]} \bigl( \mathcal E_{[1 ,2]}\bigl( \cdots \mathcal E_{[n  ,n+1]}\bigl(  h_{n+1}a\otimes I\bigr)\cdots I\bigr)\cdots
\bigr)\bigr)\bigr)\bigr),\\[2mm]
&=&\rho_0\bigl( \mathcal E_{[0 ,1]} \bigl( \mathcal E_{[1 ,2]}\bigl( \cdots \mathcal E_{[n  ,n+1]}\bigl(  h^{1/2}_{n+1}a  h^{1/2}_{n+1}\otimes I\bigr)\cdots I\bigr)\cdots
\bigr)\bigr)\bigr)\bigr),\\[2mm]
&=&\varphi_{n}^{(\rho_0, h)}(a).
\end{eqnarray*}
 This completes the proof.
\end{proof}

\begin{theorem}
Let $ K_{\{x\}\cup S(x)}\in \mathcal{B}_{\{x\}\cup S(x)}$ be given as above. Assume that the boundary condition $ (\rho_o, h= (h_u)_{u\in L})$ satisfies the following conditions
\begin{equation}\label{rho_0(h_0)}
\rho_o(h_o) = 1
\end{equation}
\begin{equation}\label{comp_h_x}
  \Tr_{\mathcal B_x]}\Bigl(K_{\{x\}\cup S(x)}^\ast  I^{(x)}\otimes\bigotimes_{y\in S(x)}h^{(y)}
K_{\{x\}\cup S(x)}\Bigr) = h^{(x)}
\end{equation}
  then  the sequence $\{\varphi^{(\rho_0, h)}_n\}$  satisfies the compatibility condition  \eqref{compatibility_eq}. Moreover, there exists
   a  quantum Markov chain.
\end{theorem}
\begin{proof}
 From Lemma \ref{localE_K_lem}, the maps   $\mathcal E_{[n, n+1]}$ are completely positive and localized in the sense of \eqref{localized_E}.
One finds
\begin{eqnarray*}
	\mathcal E_{[n, n+1]}(h_{n+1})& = & \Tr_{\mathcal B_{ W_n}]}\bigl(K_{[n, n+1]}^{\ast}h_nK_{[n, n+1]}\bigr),\\
& = &\bigotimes_{x\in  W_n} \Bigl(K_{\{x\}\cup S(x)}^\ast  I^{(x)}\otimes\bigotimes_{y\in S(x)}h^{(y)}
K_{\{x\}\cup S(x)}\Bigr),\\
& = &\bigotimes_{x\in  W_n} h^{(x)},\\
& = & h_n.
\end{eqnarray*}
Hence, by virtue of Lemma \ref{E(hn+1)=hn}, we get the desired assertion, which completes the proof.
\end{proof}

\section{Quantum Markov chains associated with Ising type models on the Comb graph $\mathbb N\rhd_0 \mathbb N$ }

In this and the forthcoming sections, we restrict ourselves to a semi-infinite Comb graph $\mathbb N\rhd_0 \mathbb N$ with distinguished vertex
$x^{(0)}= (0,0)$. In what follows, as usually, by $L$ we denote the set of all vertices and $E$ is the set of edges of $\mathbb N\rhd_0 \mathbb N$, respectively.
A simple coordinate structure is naturally associated to the  $\mathbb N\rhd_0 \mathbb N$.
Each vertex $x$ is identified with a pair of non-negative integers  $x\equiv (k,l)$.
Here $k$ is the component of $x$ on the horizontal axis and $l$ on the vertical one.

 We enumerate elements of $ W_{n}$ in the following way
\begin{equation}
x_{ W_n}^{(0)} = (n,0)  \quad  \, x_{ W_n}^{(1)} = (n-1, 1) \quad \, \cdots \quad \,   x_{ W_n}^{(n)} = (0,n)
\end{equation}
and we write
$$
\overrightarrow W_n := \bigl\{ x_{ W_n}^{(0)}, x_{ W_n}^{(1)}, \cdots, x_{ W_n}^{(n)} \bigr\}.
$$

Accordingly, we distinguish two different types of vertices: ones  having two  direct successors and others having only one. Define
 \begin{equation}\label{L1}
 L_1 = \{x\in L \, : \, |S(x)| = 1\}
 \end{equation}
 \begin{equation}\label{L2}
 L_2 = \{x\in L \, : \, |S(x)| = 2\}  = \{ x_{ W_n}^{(0)} \ : \  n\in\mathbb N\}
 \end{equation}
 here as before, $S(x)$ denotes the direct successors of $x$.

It is clear that  $L= L_1\cup L_2$. For   $u\in L_2$  one has
$$
S(u) = \{u+e_1 ,  \,  u+e_2  \}
$$
whereas, elements of $L_1$ have a form $v = (k,l)$ with $l\ge 1$ and
 $$
S(v) = \{v+e_2 \}
$$
where $e_1= (1,0)$ and $e_2= (0,1)$.

Denote
\begin{equation}\label{pauli} \id^{(u)}=\left(
          \begin{array}{cc}
            1 & 0 \\
            0 & 1 \\
          \end{array}
        \right), \
        \,        \
\sigma_z^{(u)}= \left(
          \begin{array}{cc}
            1 & 0 \\
            0 & -1 \\
          \end{array}
        \right).
\end{equation}

In what follows, we consider the
same $C^{*}$-algebra $\mathcal B_L$ but with $\mathcal B_{u}=M_{2}(\mathbb C)$
for all $u\in L$.

Let $h^{(u)} = \left(
\begin{array}{cc}
  h^{(u)}_{11} & h^{(u)}_{12} \\
  h^{(u)}_{21} & h^{(u)}_{22} \\
  \end{array}
  \right)\in M_2(\mathbb C)^+$. The family $\{h^{(u)}, \, u\in L\}$ is said to be homogeneous (or translation invariant) boundary
   conditions if $h^{(u)} = h^{(x^{0})}$ for each $u\in L$. We are interested in the resolution of the equations
\begin{equation}\label{Dysys}
    \mathrm{Tr}\left(K_{\{u\}\cup \overrightarrow{S}(u)}^\ast \id^{(u)}\otimes h_{S(u)}K_{\{u\}\cup \overrightarrow{S}(u)}\right) = h^{(u)},
\end{equation}
where
$$
\left\{
  \begin{array}{cc}
    K_{\{u\}\cup S(u)} \in \mathcal{B}_{u}\otimes\mathcal{B}_{u+e_2}; \quad h_{S(u)} = h^{(u)}\otimes h^{(u+e_2)} & \hbox{if $u\in L_1$;} \\
\\
   K_{\{u\}\cup S(u)} \in \mathcal{B}_{u}\otimes\mathcal{B}_{u+e_1}\otimes\mathcal{B}_{u+e_2}; \quad h_{S(u)} = h^{(u)}\otimes h^{(u+e_1)}\otimes h^{(u+e_2)} & \hbox{if $u\in L_2$.} \\
  \end{array}
\right.
$$

In this paper, we restrict ourselves to the description of translation-invariant solutions of \eqref{comp_h_x}.
Therefore, we always assume that: $ h^{x}=h$ for all $x\in L$, here
$$
h= \left(
\begin{array}{cc}
            h_{11} & h_{12} \\
            h_{21} & h_{22} \\
          \end{array}
\right).$$

\subsection{An Ising model on  vertices with degree two}\label{Sb_sect_Ising_L1}

In this subsection, we investigate a concrete Ising type model associated on vertices  $v\in L_1$ and their successors.

Define
\begin{equation}\label{KuL_1}
 \tilde{K}_{<v, v+e_2>} = \cos(\beta) \id^{(v)}\otimes\id^{(v+e_2)} - \sin(\beta) \sigma_z^{(v)}\otimes \sigma_z^{(v+e_2)}
\end{equation}
Put
$$
K_{\{v\}\cup S(u)} =  \tilde{K}_{<v, v+e_2>}.
$$

Then \eqref{Dysys} reduces to
\begin{equation}\label{DySyL1}
  \mathrm{Tr}\left(\tilde{K}_{<u, u+e_2>} \id^{(u)}\otimes h^{(u+e_2)}\tilde{K}_{<u, u+e_2>}\right) = h^{(u)}
\end{equation}
One can calculate
\begin{eqnarray*}
&&K_{\{u\}\cup S(u)} \id^{(u)}\otimes h^{(u+e_2)}K_{\{u\}\cup S(u)}^\ast \\
&&= \cos^{(2)}(\beta)\id^{(u)}\otimes h^{(u+e_2)} - \cos(\beta)\sin(\beta) \sigma_z^{(u)}\otimes h^{(u+e_2)}\sigma_z^{(u+e_2)}\\
&&- \cos(\beta)\sin(\beta) \sigma_z^{(u)}\otimes \sigma_z^{(u+e_2)}h^{(u+e_2)} + \sin^2(\beta) \id^{(u)}\otimes \sigma_zh^{(u+e_2)}\sigma_z.
\end{eqnarray*}

Hence, \eqref{DySyL1} is equivalent to
 $$
 \left\{
   \begin{array}{ll}
       \left( \frac{h^{(u+e_2)}_{11}+ h^{(u+e_2)}_{22}}{2}\right) -  \sin(2\beta) \left( \frac{h^{(u+e_2)}_{11}- h^{(u+e_2)}_{22}}{2}\right) = h^{(u)}_{11}; \\
 \\
    \left( \frac{h^{(u+e_2)}_{11}+ h^{(u+e_2)}_{22}}{2}\right) +  \sin(2\beta) \left( \frac{h^{(u+e_2)}_{11}- h^{(u+e_2)}_{22}}{2}\right) = h^{(u)}_{22}; \\
\\
h^{(u)}_{12} = h^{(u)}_{21} = 0.
   \end{array}
 \right.
 $$
According to the positivity of $h$, a simple calculation leads to the solution
\begin{equation}\label{solDysys1}
h^{(u)} =  a\id^{(u)}, \quad  a>0.
\end{equation}

\subsection{An Ising model on vertices with degree three}\label{Sb_sect_Ising_L2}

In this subsection, we consider another type of Ising model on triples $(v,  v+e_1, v+ e_2)$ for each $v\in L_2$.

 Define nearest neighbors interactions by
\begin{equation}\label{1Kxy1}
K_{<v,(v+e_i)>}=\exp\{\b H_{<v,(v+e_i)>}\}, \ \ i=1,2,\ \b>0,
\end{equation}
where
\begin{eqnarray}\label{1Hxy1}
H_{<v,(v,i)>}=\frac{1}{2}\big(\id^{v)}\id^{(v+e_i)}+\sigma_{z}^{(v)}\sigma_{z}^{(v+e_i)}\big),
\end{eqnarray}
and the one level nearest neighbours interaction between $v+e_1$ and $v+e_2$ by
\begin{equation}\label{Lve1e2}
  L_{>v+e_1, v+e_2<}=\exp\{J\beta H_{>v+e_1, v+e_2<}\}, \ \ J>0,
 \end{equation}
where
\begin{equation}\label{Hve1e2}
H_{>v+e_1, v+e_2<}=\frac{1}{2}\big(\id^{v+e_1}\id^{v+e_2}+\sigma_{z}^{v+e_1}\sigma_{z}^{v+e_2}\big).
\end{equation}
The constant $J>0$ is known in the physics literature as coupling constant.

The defined model is called  an {\it Ising model with competing
interactions} per vertices  $(v,(v+e_1),(v+e_2))$.

One can check that
\begin{eqnarray*}
&&H_{<u,u+e_i>}^{m}=H_{<u,u +e_i>}=\frac{1}{2}\big(\id^{(u)}\id^{(u +e_i)}+\sigma_{z}^{(u)}\sigma_{z}^{(u +e_i)}\big),\\[2mm]
\label{1Hxym}
&&H_{>u +e_1,u +e_2<}^{m}=H_{>u +e_1,u +e_2<}=\frac{1}{2}\big(\id^{(u +e_1)}\id^{(u +e_2)}+\sigma_{z}^{(u +e_1)}\sigma_{z}^{(u +e_2)}\big).
\end{eqnarray*}
Therefore,
\begin{eqnarray}
\label{Kuv_L2}&&K_{<u,u +e_i>}=K_0\id^{(u)}\id^{(u +e_i)}+K_3\sigma_{z}^{(u)}\sigma_{z}^{(u +e_i)},\\[2mm]
\label{Luv_L2}&&L_{>u +e_1, u +e_2<} = R_0\id^{(u +e_1)}\id^{(u +e_3)}+R_3\sigma_{z}^{(u +e_1)}\sigma_{z}^{(u +e_2)},
\end{eqnarray}
here
\begin{eqnarray*}
&&K_0=\frac{\exp{\beta}+1}{2},\ \ \   K_3=\frac{\exp{\beta}-1}{2},\\[2mm]
&&R_0=\frac{\exp{(J\beta)}+1}{2}, \ \ \
R_3=\frac{\exp{(J\beta)}-1}{2}.
\end{eqnarray*}
For each $u\in L_2$, we define
\begin{equation}\label{KvS(v)L2}
  K_{\{v\}\cup S(v)} = K_{<v, v+e_1>}K_{<v, v+e_2>}L_{>v+e_1, v+e_2<}.
\end{equation}

Hence,  \eqref{Dysys} reduces to
\begin{equation}\label{Dysys2}
  \mathrm{Tr}_{u]}\left(K_{\{v\}\cup S(v)}^{*} \id^{(u)}\otimes h^{(u+ e_1)}\otimes h^{(u+ e_2)} K_{\{v\}\cup S(v)}\right) = h^{(u)}.
\end{equation}

From \eqref{Kuv_L2} and \eqref{Luv_L2}, it follows that
\begin{eqnarray}\label{Ax}
K_{\{v\}\cup S(v)}&=&A\id^{(v)}\otimes\id^{(v+e_1)}\otimes\id^{(v+e_2)}+B\sigma_{z}^{(v)}\otimes\sigma_{z}^{(v+e_1)}\otimes\id^{(v+e_2)}\nonumber\\[2mm]
&&+
B\sigma_{z}^{(v)}\otimes\id^{(v+e_1)}\otimes\sigma_{z}^{(v+e_2)}+C \id^{(v)}\otimes\sigma_{z}^{(v+e_1)}\otimes\sigma_{z}^{(v+e_2)},
\end{eqnarray}
where
\begin{equation}
\left\{
  \begin{array}{ll}
    A =K_{0}^{2}R_{0}+K_{3}^{2}.R_{3}=\frac{1}{4}[\exp{(J+2)\beta}+\exp{J\beta}+2\exp{\beta}], \\
    \\
  B =K_{0}K_{3}(R_{0}+R_{3})=\frac{1}{4}\exp{J\beta}[\exp{2\beta}-1],  \\
    \\
    C =K_{0}^{2}R_{3}+K_{3}^{2}R_{0}=\frac{1}{4}[\exp{(J+2)\beta}+\exp{J\beta}-2\exp{\beta}]. \\
  \end{array}
\right.
\end{equation}

Then
\begin{eqnarray*}\label{AHHA}
&& K_{\{v\}\cup S(v)}^\ast\times[\id^{(v)}\otimes h^{(v+e_1)}\otimes h^{(v+e_2)}]\times K_{\{v\}\cup S(v)} \nonumber\\
\\
&=&\big[A^2\id\otimes h\otimes h + B^2\id\otimes \sigma_{z} h\sigma_{z} \otimes h + B^2\id\otimes h \otimes \sigma_{z} h \sigma_{z} +
C^2\id\otimes\sigma_{z} h \sigma_{z}\otimes\sigma_{z} h \sigma_{z}\big]\nonumber\\
\\
&+&\big[AC\id\otimes h \sigma_{z}\otimes h \sigma_{z} + AC\id\otimes\sigma_{z} h\otimes
\sigma_{z} h + B^2\id\otimes\sigma_{z} h \otimes h \sigma_{z} + B^2\id\otimes h \sigma_{z}\otimes\sigma_{z} h\big]\nonumber\\
\\
& +&\big[AB\sigma_{z}\otimes h\sigma_{z} \otimes h+ AB\sigma_{z}\otimes h\otimes h\sigma_{z} + AB\sigma_{z}\otimes h\otimes\sigma_{z} h +
AB \sigma_{z}\otimes\sigma_{z} h\otimes h\big] \nonumber\\
 \\
&+&\big[BC\sigma_{z}\otimes\sigma_{z} h \sigma_{z}\otimes h \sigma_{z} + BC\sigma_{z}\otimes h\sigma_{z}\otimes\sigma_{z} h \sigma_{z} + BC\sigma_{z}\otimes\sigma_{z} h\sigma_{z}\otimes\sigma_{z} h
 + BC\sigma_{z}\otimes \sigma_{z} h \otimes\sigma_{z} h \sigma_{z}\big].\nonumber\\
\end{eqnarray*}
Therefore, due to the last equality, we rewrite \eqref{Dysys2} as
follows
\begin{eqnarray} \label{eqder}
h^{(v)}&=&\Tr_{x]}K_{\{v\}\cup S(v)}^{*}\left(\id^{(v)}\otimes h\otimes h\right)K_{\{v\}\cup S(v)}\nonumber\\
&=&\left(\tau_1{\Tr( h)}^2+\tau_2{\Tr(\sigma_{z} h)}^2\right)\id^{(v)} + \tau_3\Tr(h)\Tr(\sigma_{z} h)\sigma_{z}^{(v)},
\end{eqnarray}
where $\theta=\exp(2\beta)>0$ and
\begin{equation}\label{tau123}
\left\{
   \begin{array}{ll}
     \tau_1:=A^2+2B^2+C^2= \frac{1}{4}[\theta^J(\theta^2+1)+2\theta],\\
      \\
     \tau_2:=2(AC + B^2)= \frac{1}{4}[\theta^J(\theta^2+1)-2\theta],\\
      \\
     \tau_3:=4B(A+C)=\frac{1}{2}\theta^J(\theta^2-1).
   \end{array}
 \right.\end{equation}
From
$$\Tr(h)=\frac{h_{11}+h_{22}}{2} \quad ;  \quad
\Tr(\sigma_{z} h)=\frac{h_{11}-h_{22}}{2}$$
the equation \eqref{eqder} reduces to
\begin{equation}\label{EQ1}
\left\{
   \begin{array}{lll}
\Tr(h)=\tau_1\Tr(h)^2+\tau_2\Tr(\sigma_{z} h)^2,\\
\\
\Tr(\sigma_{z} h)=\tau_3\Tr(h)\Tr(\sigma_{z} h),\\
\\
h_{21}=0,\quad  h_{12}=0.\\
   \end{array}
 \right.
 \end{equation}
Since, we are interested in transition-invariant solutions, it is convenient to find
those one with $h_{1,1}=h_{2,2}$. So, \eqref{EQ1} is reduced to
\begin{equation*}
    h_{11}=h_{22}=\frac{1}{\tau_1}.
\end{equation*} Hence,
\begin{equation}\label{solDysys2}
h^{(u)} = \frac{1}{\tau_1}\id^{(u)};\quad \forall u\in L_2.
\end{equation}

Now we are ready to formulate a main result of this section.

\begin{theorem}\label{thm1}
For the Ising  (ZZ) coupling model (\ref{KuL_1}) and the Ising model with competing interactions
(\ref{1Kxy1}), (\ref{Lve1e2}),\, $\b>0, J>0$ on the Comb graph $\mathbb{N}\rhd_0\mathbb{N}$,
  there exists a quantum Markov chain  $\varphi$ with homogeneous boundaries.
   Moreover, it  can be written on the local algebra
$\mathcal{B}_{L, loc}$ by:
\begin{equation}
\varphi(a)=\alpha^{n}\Tr\bigg(a\prod_{i=0}^{n-1}K_{[i,i+1]}K_{[i,i+1]}^{*}\bigg),
\ \ \forall a\in B_{ [0,n]}.
\end{equation}
\end{theorem}

We notice that the QMC $\varphi$ is called \textit{disordered phase} of the Ising model.


\begin{proof}
According to the Ising type models considered in the subsection \ref{Sb_sect_Ising_L1} and  the subsection \ref{Sb_sect_Ising_L2},  the equation \eqref{Dysys} admits a unique translation-invariant solution
\begin{equation}\label{h=alphaI}
h^{(u)}_{\alpha}=
 \left(
  \begin{array}{cc}
    \alpha & 0 \\
    0 & \alpha \\
  \end{array}
\right)
\end{equation}
where $\alpha=\frac{1}{\tau_1}$. The initial functional $\rho_0$ can be chosen of the form
$$
\rho_0(a) = \tau_1\Tr(a)
$$
so that $\rho_0(h^{x^{(0)}}) =\Tr(\omega_0h^{x^{(0)}})  = 1$,  where $\omega_0 = \tau_1 \id^{x^{(0)}}$\\
Let $n\in\mathbb N, a\in B_{[0,n]}$, according   (\ref{h=alphaI}) and from Lemma \ref{E(hn+1)=hn},  one finds
\begin{eqnarray*}
\varphi(a) &=& \varphi_n^{(\rho_0, h)}(a)\\
&=&  \rho_0\Bigl(\mathcal E_{[0,1]}\Bigl(a_{ W_0}\otimes\mathcal E_{[1,2]}\Bigl(a_1\otimes\cdots \otimes \mathcal E_{[n-1 ,n+1]}\Bigl(h_{n} a\Bigr)\Bigr)\Bigr)\Bigr)\\
&=&  \Tr\bigl( \omega_0\Tr_{\mathcal{B}_{W_0]}}\Bigl(K_{[0,1]}^{*}\Tr_{\mathcal{B}_{W_1]}}\Bigl(K_{[1, 2]}^{*}\cdots \otimes \Tr_{\mathcal{B}_{W_{n-1}]}}\Bigl(K_{[n-1,n]}^{*}h_{n} a K_{[n-1,n]}\Bigr)\cdots K_{[1, 2]}\Bigr)K_{[0,1]}\Bigr)\Bigr)\\
&=&  \Tr\bigg(\omega_{0}h_{n}\bigg(\prod_{m=0}^{n-1}K_{[m,m+1]}^{*} a K_{[m,m+1]}\bigg) \bigg) \nonumber\\
&=&\alpha^{| W_{n}|-1}\Tr\bigg(a \prod_{m=0}^{n-1}K_{[m,m+1]}K_{[m,m+1]}^{*}\bigg) \nonumber\\
&=&\alpha^{n}\Tr\bigg(a\prod_{i=0}^{n-1}K_{[i,i+1]}K_{[i,i+1]}^{*} \bigg).
\end{eqnarray*}
This completes the proof.
\end{proof}

\section{Clustering property}

The Comb graph $\mathbb{N} \rhd _{0}\mathbb{N}= (V, E)$ is invariant under the action of the semi-group $\mathcal{G}^+$ of translations of the form
\begin{equation}\label{J_n}
J_n: u= (k,l)\in \mathbb{N}\times\mathbb{N} \mapsto u + ne_1  = (k+n, l).
\end{equation}
For each $g= (n,0) \in L_2$ we denote
$$
\tau_g(u) := J_n(u).
$$
In these notations, we have $\mathcal{G}^+ := \{\tau_g \ : \  g \in L_3\}$.

\begin{lemma}\label{CC}
Assume that $n_0\in\mathbb{N}$. Let $a \in \mathcal{B}_{[0,n_0]}$  and $b\in\mathcal{B}_{x^{(0)}}$. Then
\begin{equation}\label{clus_prop_b_n}
\lim_{n\to +\infty}\varphi(aJ_n(b)) = \varphi(a) \varphi(b),
\end{equation}
where $ J_n(b)= b^{(x_{ W_n}^{(0)})}$.
\end{lemma}
\begin{proof}
 For large enough $n$, one gets $ab_n = a\otimes b_n$, $n>n_0$. So,
\begin{eqnarray*}
\varphi(a\otimes b_n)&=& \varphi_n^{(\rho_0, h_\alpha)}(a\otimes b_n)\\
&=& \rho_0\circ \mathcal{E}_{[0,1]}\circ\cdots\circ \mathcal{E}_{[n_0, n_0+1]}(a\otimes  \mathcal{E}_{[n_0+1, n_0+2]}(\id_{ W_{n_0+1}}\otimes \cdots\\
&&\otimes\mathcal{E}_{[n-1,n]}\circ \hat{\mathcal{E}}_{[n , n +1]}(b_n\otimes\id_{ W_{n+1}}).
\end{eqnarray*}
Here, the boundary condition is taken, as before, $h^{(u)} = h_\alpha = \alpha\id$ is a common fixed points of the systems (\ref{DySyL1}) and (\ref{Dysys2}) and the initial state $\rho(\cdot) = \mathrm{Tr}(\omega_0 \, \cdot),  \, \omega_0 = \frac{1}{\alpha}\id$ with $\alpha =\frac{4}{\theta^J(\theta^2+1)+2\theta}$.

Then, one finds
\begin{eqnarray*}
 && \hat{\mathcal{E}}_{[n,n+1]}(b_n\otimes\id_{ W_{n+1}})  =
\mathrm{Tr}_{n]}\left(  K_{[n,n+1]}^{\ast}\bh_{n+1}^{1/2} ( b\otimes \id_{ W_{n+1}})\bh_{n+1}^{1/2}K_{[n,n+1]}\right)\\
  &=&\bigotimes_{u\in W_n}\mathrm{Tr}_{u]}
\left(  K_{\{u\}\cup S(u)}^{\ast}\bigotimes_{v\in S(u)}(h^{(v)})^{1/2} ( b\otimes \id_{ W_{n+1}})\bigotimes_{v\in S(u)} (h^{(v)})^{1/2}K_{[n,n+1]}\right)\\
&=& \mathrm{Tr}_{x_{ W_n}^{(0)}]}
 \left(  K_{\{x_{ W_n}^{(0)}\}\cup S(x_{ W_n}^{(0)})}^{\ast}
 b^{(x_{ W_n}^{(0)})}\otimes h_{S(x_{ W_n}^{(0)})} K_{\{x_{ W_n}^{(0)}\}\cup S(x_{ W_n}^{(0)})}\right)\otimes\\
&&\bigotimes_{u\in W_n\setminus \{x_{ W_n}^{(0)}\}}\mathrm{Tr}_{u]}\left(  K_{\{u\}\cup S(u)}^{\ast}\bigotimes_{v\in S(u)}h^{(v)}K_{[n,n+1]}\right).
\end{eqnarray*}

By \eqref{comp_h_x}
 $$
\mathrm{Tr}_{u]}\left(  K_{\{u\}\cup S(u)}\bigotimes_{v\in S(u)}h^{(v)}K_{[n,n+1]}^{\ast}\right) = h^{(u)}
$$
and  $$\mathrm{Tr}_{x_{ W_n}^{(0)}]}
 \left(  K_{\{x_{ W_n}^{(0)}\}\cup S(x_{ W_n}^{(0)})}^{\ast}b^{(x_{ W_n}^{(0)})}
 \left(\id^{(x_{ W_n}^{(0)})}\otimes h_{S(x_{ W_n}^{(0)})} \right) K_{\{x_{ W_n}^{(0)}\}\cup S(x_{ W_n}^{(0)})}\right)\\
$$
\begin{equation}\label{trx](KbK)}
=  \alpha^2 \left[(A^2 + C^2) b^{(x_{ W_n}^{(0)})} + 2B^2  \sigma_{z} b^{(x_{ W_n}^{(0)})}\sigma_{z} \right].
\end{equation}

Then
$$
\mathcal{E}_{[n-1,n]}(\id_{ W_{n-1}}\otimes\hat{\mathcal{E}}( b^{(x_{ W_n}^{(0)})}\otimes \id))
 =\alpha^2 (A^2 + C^2)\mathcal{E}_{[n-1,n]}\left(\id_{ W_{n-1}}\otimes b^{(x_{ W_n}^{(0)})}
\otimes\bigotimes_{u\in  W_{n}\setminus \{x_{ W_n}^{(0)}\}}h^{(u)}\right)
$$
$$
  + 2\alpha^2B^2 \mathcal{E}_{[n-1,n]}\left(\id_{ W_{n-1}}\otimes \sigma_{z} b^{(x_{ W_n}^{(0)})}\sigma_{z}
\otimes\bigotimes_{u\in  W_{n}\setminus \{x_{ W_n}^{(0)}\}}h^{(u)} \right).
$$
A small calculation leads to
$$
\Tr_{x_{ W_{n-1}}^{(0)}]}\left(K_{\{x_{ W_{n-1}}^{(0)}\}\cup S(x_{ W_{n-1}}^{(0)})} \id^{(x_{ W_{n-1}}^{(0)})}\otimes b^{(x_{ W_{n-1}}^{(0)}+e_1)}
\otimes h^{(x_{ W_{n-1}}^{(0)}+e_2)}  K_{\{x_{ W_{n-1}}^{(0)}\}\cup S(x_{ W_{n-1}}^{(0)})}^*\right)
$$
$$
 = \alpha \left((A^2 + 2B + C^2)\mathrm{Tr}(b)\id^{x_{ W_{n-1}}^{(0)}} + 2B(A+C)\mathrm{Tr}(\sigma_{z} b)\sigma_{z}^{x_{ W_{n-1}}^{(0)}} \right)
$$
$$
=\Tr_{x_{ W_{n-1}}^{(0)}]}\left(K_{\{x_{ W_{n-1}}^{(0)}\}\cup S(x_{ W_{n-1}}^{(0)})}
\id^{(x_{ W_{n-1}}^{(0)})}\otimes (\sigma_{z} b\sigma_{z})^{(x_{ W_{n-1}}^{(0)}+e_1)}
\otimes h^{(x_{ W_{n-1}}^{(0)}+e_2)}  K_{\{x_{ W_{n-1}}^{(0)}\}\cup S(x_{ W_{n-1}}^{(0)})}^*\right).
$$
This implies
$$
\mathcal{E}_{[n-1,n]}(\id_{ W_{n-1}}\otimes\hat{\mathcal{E}}( b^{(x_{ W_n}^{(0)})}\otimes \id))
= \Tr(b)h_{n-1} +  \Tr(\sigma_{z} b) (\frac{\tau_3}{4\tau_1}) \sigma_{z}^{x_{ W_{n-1}}^{(0)}}h_{n-1}.
$$

On the other hand, using \eqref{trx](KbK)} one gets
$$
\varphi(b^{(o)}) =\rho_0( \mathcal{E}_{[0,1]}(b^{(o)}\otimes h_1)) = \Tr(b)
$$
Since  for each $k\in\mathbb{N}$ we obtain
$$
\Tr_{x_{ W_{k}}^{(0)}]}\left(K_{\{x_{ W_{k}}^{(0)}\}\cup S(x_{ W_{k}}^{(0)})} \id^{(x_{ W_{k}}^{(0)})}\otimes \sigma_{z}^{(x_{ W_{k}}^{(0)}+e_1)}
\otimes h^{(x_{ W_{k}}^{(0)}+e_2)}  K_{\{x_{ W_{k}}^{(0)}\}\cup S(x_{ W_{k}}^{(0)})}^*\right)
$$
$$
= \alpha(A+B)C \sigma_{z}^{(x_{ W_{k}}^{(0)})} = \frac{\tau_3}{4\tau_1}\sigma_{z}^{(x_{ W_{k}}^{(0)})}.
$$
and taking into account (\ref{comp_h_x}), a simple iteration  leads to
$$
\mathcal{E}_{[k,k+1]}(\id\otimes \cdots \mathcal{E}_{[n-1,n]} (\id\otimes \hat{\mathcal{E}}_{[n,n+1]}(b^{x_{ W_n}^{(1)}}) =
\varphi(b) h_{k}
+  \mathrm{Tr}(\sigma_{z} b)\left(\frac{\tau_3}{4\tau_{1}} \right)^{n-k}
(\sigma_{z}^{(x_{ W_k}^{(1)})}\otimes \id_{ W_k\setminus \{x_{ W_k}^{(1)}\}})h_k.
$$
In particular, by taking $k=n_0 + 1, n_0 +2$, one gets
$$
\varphi_n(aJ_{n}(b)) =
\varphi(b)\rho_0(\mathcal{E}_{[0,1]}(a_{ W_0}\otimes \mathcal{E}_{[1,2]}(a_{ W_1}\cdots
\mathcal{E}_{[n_0-1,n_0]}(a_{ W_{n_0-1}}\otimes \mathcal{E}_{[n_0,n_0+1]}(a_{\Lambda}\otimes h_{n_0+1}))))
$$
$$
+ \mathrm{Tr}(\sigma_{z} b) \left(\frac{\tau_3}{4\tau_{1}}\right)^{n-n_0-2}\rho_0(\mathcal{E}_{[0,1]}(a_{ W_0}\otimes \mathcal{E}_{[1,2]}(a_{ W_1}\cdots
\mathcal{E}_{[n_0,n_0+1]}(a_{\Lambda}\otimes\mathcal{E}_{[n_0+1,n_0]}(\sigma_{z}^{(x_{ W_{n_0+1}})}\otimes h_{n_0+2})))).
$$
Thus
\begin{equation}\label{varph(ab_nexpress)}
\varphi_n(aJ_{n}(b)) = \varphi(b)\varphi_{n_0}(a) + \mathrm{Tr}(\sigma_{z} b)  \left(\frac{\tau_3}{4\tau_{1}}\right)^{n-n_0-2}
 \varphi_{n_0+1}(a\otimes \sigma_{z}^{(x_{ W_{n_0+1}})} ).
\end{equation}
From \eqref{tau123} it follows that
$$
\frac{\tau_3}{4\tau_{1}} = \frac{\theta^J(\theta^2-1)}{2(\theta^J(\theta^2+1)+2\theta)}\in [0 , \frac{1}{2}]; \quad \theta = \exp(2\beta)>1, J>0.
$$
Now taking the limit $n\to +\infty$ in \eqref{varph(ab_nexpress)} we arrive at \eqref{clus_prop_b_n}.
\end{proof}

The main result of this section is the following.

\begin{theorem}\label{thm_clus_property}
   Let $\varphi$ be a QMC associated with the Ising with (ZZ) coupling model on the comb graph $\mathbb{N}\rhd_0\mathbb{N}$. Then

\begin{equation}\label{clust_prop}
\lim_{|g|\to +\infty}\varphi(a\tau_{g}(b)) = \varphi(a) \varphi(b)
\end{equation}
for all $a,b\in\mathcal{B}$.
 \end{theorem}

\begin{proof}

Let $a,b\in \mathcal{B}_{L, loc}$. There exist  $n,0, m_0\in\mathbb{N}$ such that
$a= a_{n_0}\in\mathcal{B}_{[0,n_0]}$ and  $b = b_{m_0} \in  \mathcal{B}_{[0, m_0]}$. Without lost of generality, we may assume that  $b$ is localized in the following form
$$
b =\bigotimes_{u\in [0,m_0]}b_u =  \bigotimes_{k=0}^{m_0}b_{ W_k}; \quad b_{ W_k} = \bigotimes_{u\in  W_{k}}b_u \in \mathcal{B}_{ W_k}
$$
One can see that
$$
J_n(b)  = \bigotimes_{u\in  [0,m_0] }b_{u}^{(u + ne_1)} = \bigotimes_{v\in [n, n +m_0]}\tilde{b}_{v} \in\mathcal{B}_{ [n, n +m_0] }
$$
where
$$
\tilde{b}_{v}  = \left\{
                   \begin{array}{ll}
                     b_{v-ne_1}, & \hbox{if $v-ne_1\in  \Lambda_{m_0}$;} \\
                     1, & \hbox{otherwise.}
                   \end{array}
                 \right.
$$

 Define
 \begin{equation}\label{fJn(lambda)}
F_{J_n( W_{k })}(a):= \bigotimes_{v\in J_n( W_{k})}\Tr_{v]}\left(K_{\{v\}\cup S(v)} a  K^{*}_{\{v\}\cup S(v)}\right); \quad a\in \mathcal{B}_{[k,k+1]}
\end{equation}


Taking into account (\ref{comp_h_x}), one finds

$$ \hat{\mathcal{E}}_{[n+m_0, n+m_0+1]}(\tilde{b}_{ W_{n+m_0}})))
= \left(\bigotimes_{v\in J_n( W_{m_0})}
\Tr_{v]}\left(K_{\{v\}\cup S(v)}  b_{v} h_{S(v)} K^{*}_{\{v\}\cup S(v)}\right)\right)\otimes
\left(\bigotimes_{v\in   W_{n+m_0}\setminus J_n( W_{m_0})} h_v\right)
$$
$$
= \alpha^{ | W_{m_0+1}|} F_{J_n( W_{m_0})}(\tilde{b}_{J_n( W_{m_0})})\otimes
\left(\bigotimes_{v\in   W_{n+m_0}\setminus J_n( W_{m_0})} h_v \right).
$$
According to the structure of the  comb graph $\mathbb{N}\rhd_0\mathbb{N}$ one has
$$
J_n( W_{k+1}) = \bigcup_{v\in J_n( W_k)}S(v).
$$

Then
\begin{eqnarray*}
&&{\mathcal{E}}_{[n+m_0-1, n+m_0 ]}(\tilde{b}_{ W_{n+m_0-1}}\otimes\hat{\mathcal{E}}_{[n+m_0, n+m_0+1]}(\tilde{b}_{ W_{n+m_0}}))) )=\\
&&=\alpha^{ | W_{m_0+1}|}
\bigotimes_{u\in J_n( W_{m_0-1})}\Tr_{u]}\left(K_{\{u\}\cup S(u)}  \tilde{b}_u\otimes   F_{J_n( W_{m_0})}(\tilde{b}_{J_n( W_{m_0})})K_{\{u\}\cup S(u)} \right)\\
&&\otimes\bigotimes_{w\in   W_{n+m_0-1}\setminus J_n( W_{m_0-1})} h_w, \\
&&=\alpha^{ | W_{m_0+1}|}
    F_{J_n( W_{m_0-1})}\left(\tilde{b}_{J_n( W_{m_0-1})}\otimes F_{J_n( W_{m_0})}(\tilde{b}_{J_n( W_{m_0})}) \right)
\otimes\bigotimes_{w\in   W_{n+m_0-1}\setminus J_n( W_{m_0-1})} h_w.
\end{eqnarray*}
An iteration leads to
$$
\mathcal{E}_{[n,n+1]}(\tilde{b}_{ W_n}\otimes \cdots \mathcal{E}_{[n+m_0-1, n+m_0]}(\tilde{b}_{ W_{n+m_0-1}}\otimes \hat{\mathcal{E}}_{[n+m_0, n+m_0]}(\tilde{b}_{n+m_0})))
$$
$$
= \alpha^{ | W_{m_0+1}|}
     F_{J_n( W_{0})}\left(\tilde{b}_{J_n( W_{0})}\otimes\cdots F_{J_n( W_{m_0-1})}\left(\tilde{b}_{J_n( W_{m_0-1})}\otimes F_{J_n( W_{m_0})}(\tilde{b}_{J_n( W_{m_0})}) \right)\right)
\otimes\bigotimes_{w\in   W_{n}\setminus J_n( W_{0})} h_w.
$$
Since $J_{n}( W_0) = \{x_{ W_n}^{(0)}\}$ and $h_w = \alpha\id$ for each $w$ then
\begin{eqnarray*}
&&\mathcal{E}_{[n,n+1]}(\tilde{b}_{ W_n}\otimes \cdots \mathcal{E}_{[n+m_0-1, n+m_0]}(\tilde{b}_{ W_{n+m_0-1}}\otimes \hat{\mathcal{E}}_{[n+m_0, n+m_0]}(\tilde{b}_{n+m_0})))\\
&&= \alpha^{ | W_{m_0+1}|+ | W_n|-1}
     F_{x_{ W_{n}}^{(0)}}\left(\tilde{b}_{J_n( W_{0})}\otimes\cdots F_{J_n( W_{m_0-1})}\left(\tilde{b}_{J_n( W_{m_0-1})}\otimes F_{J_n( W_{m_0})}(\tilde{b}_{J_n( W_{m_0})}) \right)\right)
\otimes\id_{ W_{n}\setminus \{x_{ W_{n}}^{(0)}\}},\\
&& = b_n \otimes h_{n+1}
\end{eqnarray*}
where
$$
g_n: = \alpha^{ |m_0|}F_{x_{ W_{n}}^{(0)}}\left(\tilde{b}_{J_n( W_{0})}\otimes\cdots F_{J_n( W_{m_0-1})}\left(\tilde{b}_{J_n( W_{m_0-1})}\otimes F_{J_n( W_{m_0})}(\tilde{b}_{J_n( W_{m_0})}) \right)\right)\in \mathcal{B}_{x_{ W_n}^{(0)}}.
$$

One can easily check that $b_n =\alpha^{ |m_0|} J_n(b)$ with
$$
g =F_{ W_0}\left(b_{ W_{0}}\otimes\cdots F_{  W_{m_0-1}}\left( b_{ W_{m_0-1}}\otimes F_{ W_{m_0}}( b_{  W_{m_0} }) \right)\right)\in \mathcal{B}_{ W_{0}}.
$$
Then
\begin{eqnarray*}
&&\varphi(a_{n_0}\otimes J_n(b_{m_0}) = \rho_0(\mathcal{E}_{[0,1]}(a_{ W_0}\cdots \mathcal{E}_{[n_0,n_0+1]}(a_{ W_{n_0}}\otimes \mathcal{E}_{[n_0,n_0+1]}(\id_{ W_{n_0+1}}\otimes \cdots\\
&& \mathcal{E}_{[n,n+1]}(\tilde{b}_{ W_n}\otimes \cdots \mathcal{E}_{[n+m_0-1, n+m_0]}(\tilde{b}_{ W_{n+m_0-1}}\otimes \hat{\mathcal{E}}_{[n+m_0, n+m_0]}(\tilde{b}_{n+m_0})))))))
\\
&&= \rho_0(\mathcal{E}_{[0,1]}(a_{ W_0}\cdots \mathcal{E}_{[n_0,n_0+1]}(a_{ W_{n_0}}\otimes \mathcal{E}_{[n_0,n_0+1]}(\id_{ W_{n_0+1}}\otimes \cdots \mathcal{E}_{[ n,n+1]}(g_{n}\otimes h_{n+1})))))\\
&&= \alpha^{ |m_0|}\varphi_n(a_{n_0}J_n(g))
\end{eqnarray*}
By Lemma \ref{CC}, one finds
$$
  \lim_{n\to +\infty}\varphi_n(a_{n_0}J_n(g)) =   \varphi(a_{n_0} ) \varphi(g).
$$
Due to $F_{ W_k}= \mathcal{E}_{[k,k+1]}$ for $0\le k\le m_0$, we obtain
$$
   \alpha^{ |m_0|}\varphi(g) =  \varphi(b_{m_0}).
$$
Therefore,
$$
\lim_{n\to+\infty} \varphi(a_{n_0}J_n(b_{m_0})) = \varphi(a_{n_0})\varphi(b_{m_0})
$$
which completes the proof.
\end{proof}

\section*{Acknowledgments}
The authors gratefully acknowledge Qassim University, represented by the Deanship of Scientific Research, on the
financial support for this research under the number (cba-2019-2-2-I-5400)
during the academic year 1440 AH / 2019 AD.

\end{document}